\documentclass[twocolumn]{article} 
\usepackage[scale=0.8]{geometry}

\usepackage{authblk}  
\usepackage{natbib}
\usepackage{times}  
\usepackage{helvet}  
\usepackage{courier}  
\usepackage[hyphens]{url}  
\usepackage{graphicx} 
\usepackage{natbib}  
\usepackage{caption} 
\usepackage{algorithm}
\usepackage{algorithmic}

\usepackage{amsmath,amssymb,amsthm}
\usepackage[capitalize]{cleveref}
\usepackage{tcolorbox}
\usepackage{booktabs}
\usepackage{thm-restate}
\usepackage{complexity}
\newclass{\paraNP}{paraNP}

\newtheorem{theorem}{Theorem}
\newtheorem{corollary}[theorem]{Corollary}

\usepackage{newfloat}
\usepackage{listings}
\DeclareCaptionStyle{ruled}{labelfont=normalfont,labelsep=colon,strut=off} 
\floatstyle{ruled}
\newfloat{listing}{tb}{lst}{}
\floatname{listing}{Listing}

\newcommand{\defproblem}[3]{%
      \noindent
      \begin{center}
      \begin{tcolorbox}[title=#1,left=0mm,top=0mm,bottom=0mm,right=0mm,boxsep=1mm]
        \begin{tabular}{@{}ll}
            \textbf{Input:} & \parbox[t]{.74\columnwidth}{#2}\\
            \addlinespace
            \textbf{Question:} & \parbox[t]{.74\columnwidth}{#3}\\
        \end{tabular}
      \end{tcolorbox}
      \end{center}
      }

\newcommand{\prob}{{\sc HomNC}}
\newcommand{\probtwo}{{\sc HetNC}}
\newcommand{\hetprobbin}{{\sc HetNC-B}}
\newcommand{\hetprobun}{{\sc HetNC-U}}
\newcommand{\OO}{\mathcal{O}}
\newcommand{\mcap}{K}
\newcommand{\capa}{\kappa}
\newcommand{\like}{\lambda}
\newcommand{\caches}{C}
\newcommand{\users}{U}
\newcommand{\catalog}{S}
\newcommand{\size}{\sigma}
\newcommand{\score}{\ell}
\newcommand{\vc}{\mathtt{vc}}

\title{Parameterized Complexity of Caching in Networks\thanks{An extended abstract of this paper will appear in the proceedings of AAAI 2025.}}
\author[1]{Robert Ganian}
\author[2]{Fionn Mc Inerney}
\author[2]{Dimitra Tsigkari}
\affil[1]{Algorithms and Complexity Group, TU Wien, Vienna, Austria}
\affil[2]{Telef\'{o}nica Scientific Research, Barcelona, Spain}

\date{}

\begin{document}

\maketitle

\begin{abstract}
The fundamental caching problem in networks asks to find an allocation of contents to a network of caches with the aim of maximizing the cache hit rate.
Despite the problem's importance to a variety of research areas---including not only content delivery, but also edge intelligence and inference---and the extensive body of work on empirical aspects of caching, very little is known about the exact boundaries of tractability for the problem beyond its general \NP-hardness. 
We close this gap by performing a comprehensive complexity-theoretic analysis of the problem through the lens of the parameterized complexity paradigm, which is designed to provide more precise statements regarding algorithmic tractability than classical complexity. Our results include algorithmic lower and upper bounds which together establish the conditions under which the caching problem becomes tractable.
\end{abstract}

\section{Introduction} \label{sec:intro}
Caching is one of the most important tasks that need to be handled by modern-day content delivery networks~(CDNs), and its role is expected to grow even further in the future, e.g., in the context of wireless communications~\citep{liu2016caching,paschos2018role}
and artificial intelligence. 
For the latter, it is crucial to cache: trained models for inference requests~\citep{salem2023toward, zhu2024towards, yu2024accelerating}, information that can accelerate the training of large models~\citep{ lindgren2021efficient, zhang2024cached}, and trained model parts in distributed learning paradigms~\citep{thapa2022splitfed, tirana2024workflow}.
Typically, in a caching network, one needs to decide where to store the contents in order to maximize metrics related to, e.g., network performance or user experience, with cache hit rate being the most predominant one~\citep{paschos2020cache}.
While many recent works have investigated dynamic caching policies that either decide on caching before time-slotted requests or via the eviction policy~\citep{paschos2019learning,bhattacharjee2020fundamental,rohatgi2020near,paria2021texttt, mhaisen2022optimistic}, in this article we focus on proactive caching, which decides how to fill the caches based on anticipated requests. Proactive caching is of crucial importance in edge caching~\citep{bastug2014living, tadrous2015optimal} and is applied in contemporary caching architectures like Netflix's CDN~\citep{netflix_fill_patterns}.
In practice, content requests in video streaming services can be steered through trends or recommendation systems, leading to accurate predictions on the content popularity~\citep{RecImpact-IMC10,netflix_and_fill}. 
  
In the proactive setting, the caching problem in a network of caches~(with possibly overlapping coverage) can be naturally represented as a bipartite graph between the set of users and the network of caches (of limited capacities), a set of contents, and information on the users' anticipated requests; the task is to allocate the contents to caches to maximize the cache hit rate.
This problem was proven to be \NP-hard in the seminal work that formalized this setting~\citep{femto}.
The subsequent numerous variants of this base problem~\citep{poularakis2014approximation,dehghan2016complexity,krolikowski2018decomposition,ricardo2021caching,tsigkari2022approximation} were all shown to be \NP-hard by a reduction from it.

In the above works, the theoretical intractability was typically overcome via the use of heuristics or approximation algorithms. In fact, despite the extensive body of work on the caching problem, our understanding of its foundational complexity-theoretic aspects is still in its infancy.
Specifically, very little is known about the problem's ``boundaries of tractability'', i.e., the precise conditions under which it becomes tractable.
This gap is in stark contrast with recent substantial advances made in this direction for problems arising in artificial intelligence and machine learning such as 
Bayesian network learning~\citep{OrdyniakS13,GanianK21,gruttemeier2022learning},
 data completion~\citep{GanianKOS18,GanianKOS20,EibenGKOS21}, 
and resource allocation~\citep{BliemBN16,DeligkasEGHO21,EibenGHO23}.

The aim of this article is to close the aforementioned gap by carrying out a comprehensive study of the complexity-theoretic aspects of caching in its modern network-based formalization through the lens of \emph{parameterized complexity}~\citep{DowneyF13,paramcompbook}, which provides more refined tools to identify the conditions under which an \NP-hard problem becomes tractable when compared to classical complexity. In the parameterized setting, the running times of algorithms are studied with respect to the input size $n$, but also an integer parameter~$k$, which intuitively captures well-defined structural properties of the input. The most desirable outcome here is \emph{fixed-parameter tractability} (\FPT), which means that the problem can be solved in time $f(k)\cdot n^{\OO(1)}$ for a computable function~$f$.
A weaker notion of this is \XP-tractability, which means that the problem can be solved in time $n^{f(k)}$.
The least favorable outcome is $\paraNP$-hardness, which means that the problem remains \NP-hard even when $k$ is a constant. 
Lastly, excluding fixed-parameter tractability is done by proving the problem is \W[1]-\emph{hard} via a parameterized reduction.

\subparagraph{Our Contributions.} \label{subsec:contributions}
We consider three fundamental variants of the \textsc{Network-Caching} problem, which differ solely on the modeling and encoding of the content sizes.
The main challenges of our investigation arise from the various aspects of the problem: when solving an instance one needs to (1) adhere to (knapsack-like) capacity constraints while (2) taking into account the graph structure of the caching network, and also (3) deal with the large variety of naturally occurring parameterizations.
In particular, we study the complexity of these variants with respect to six parameters: 
the number~$\caches$ of caches, 
the maximum capacity~$\mcap$ of any cache, 
the number~$\catalog$ of contents, 
the number~$\users$ of users, 
the maximum degree~$\Delta$ of the network, and
the maximum number~$\like$ of contents that any user may request.\footnote{The problem definitions and motivation for the employed parameterizations are deferred to Section~\ref{sec:prelims}.} 

We first consider the most computationally challenging variant of the problem, where the content and cache sizes are encoded in their ``exact form'' (as binary-encoded numbers).
Interestingly, the inherent hardness of this variant---denoted \hetprobbin---makes it easier to perform a detailed analysis compared to the two other variants.
Note that the parameter $\mcap$ is incompatible with this model (bounding the cache capacities translates the problem to the unary-encoded setting investigated later).
We begin by establishing the problem is \FPT\ under the combined parameterizations\footnote{Parameterizing by a set of parameters is equivalent to parameterizing by their sum~\citep{paramcompbook}.} of $\caches+\catalog$ (Theorem~\ref{thm:C+S}) and $\users+\catalog$ (or the complexity-theoretically equivalent $\users+\lambda$) (Corollary~\ref{cor:U+S}); as \hetprobbin\ is the most general of the three models, these tractability results carry over to the other two. Then, we show that \hetprobbin\ is \paraNP-hard under all remaining parameterizations (Theorems~\ref{thm:hard-NAE-sat}~and~\ref{thm:knapsack} and Corollary~\ref{cor:split-users-knapsack}).
The clean cut between tractability and intractability for \hetprobbin\ is depicted in \cref{fig:hetero_binary_landscape}~(top).

\begin{figure}[!t]
\centering  
\includegraphics[width=\columnwidth]{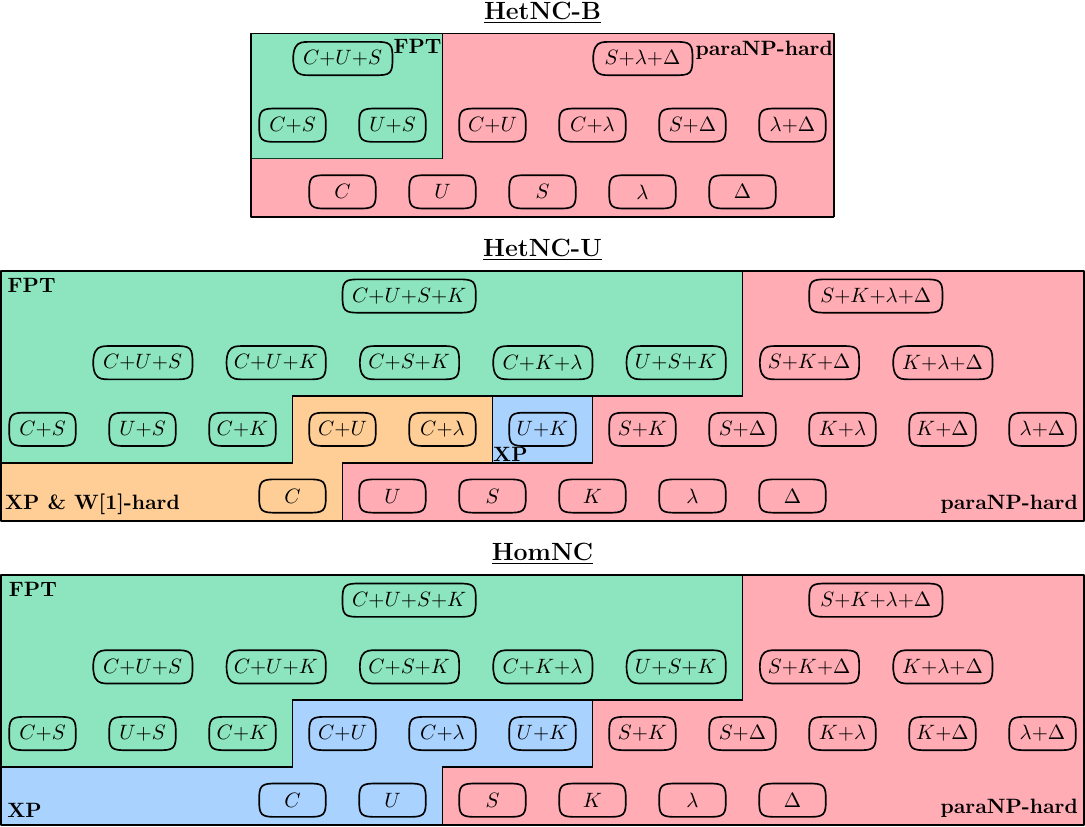}
\caption{Complexity landscapes of \hetprobbin\ (top), \hetprobun\ (middle), and \prob\ (bottom). We (mostly) omit parameterizations including $\caches+\Delta$, $\users+\Delta$, $\users+\lambda$, or $\catalog+\lambda$, as the first two, the third, and the last are complexity-theoretically equivalent to $\caches+\users$, $\users+\catalog$, and $\catalog$, respectively. Indeed, for any pair of combined parameters claimed to be equivalent, each of them can be bounded by a computable function of the other. To prove the first two equivalences, note that $\Delta\leq \caches+\users$, $\users\leq \caches\cdot \Delta$, and $\caches \leq \users\cdot \Delta$. For the latter two, note that $\catalog\leq \users \cdot \like$ and $\lambda\leq \catalog$.}
\label{fig:hetero_binary_landscape}
\end{figure}

Next, we consider the unary variant of the problem, where one still distinguishes between contents of different sizes, but only on a more categorical scale (e.g., by distinguishing between full movies, longer shows, and shorter shows in the setting of video streaming).
Here, the sizes of the contents are encoded in unary, which affects the complexity of the associated \hetprobun\ problem.
On the positive side, in addition to $\caches+\catalog$ and $\users+\catalog$, we identify a third fixed-parameter tractable fragment of the problem, specifically $\caches+\mcap$ (Theorem~\ref{thm:xp-c}).
Moreover, we obtain \XP-tractability for \hetprobun\ with respect to $\caches$ alone (Theorem~\ref{thm:xp-c}), and complement this result by establishing \W\textup{[1]}-hardness (and hence, ruling out fixed-parameter tractability) even when combining $\caches$ with either $\users$ or $\like$ (Theorem~\ref{thm:one-user} and Corollary~\ref{cor:split-users}).
All fragments not covered by these cases are then \paraNP-hard (Theorems~\ref{thm:hard-NAE-sat}~and~\ref{thm:one-user}), except for $\users+\mcap$ for which we establish \XP-tractability (Theorem~\ref{thm:xp-ud}).
The obtained complexity landscape of \hetprobun\ is illustrated in \cref{fig:hetero_binary_landscape}~(middle). 

Finally, we turn to the simplest, totally uniform case where each content has the same size; we denote the associated problem \prob.
We first demonstrate the challenging nature of even this variant by strengthening the previously known lower bound~\citep{femto} and establishing \NP-hardness when $\catalog=\like=2$, $\mcap=1$, and, additionally, $\Delta=3$ (Theorem~\ref{thm:hard-NAE-sat}).
Furthermore, unlike for \hetprobun, here we establish \XP-tractability with respect to $\users$ alone (Corollary~\ref{cor:xp-u}).
For the lower bounds, all \paraNP-hardness results, other than with respect to $\users$ alone, can be shown to carry over to the simpler setting of \prob.
Our results for \prob\ are displayed in \cref{fig:hetero_binary_landscape}~(bottom).

To complete the complexity landscapes for \hetprobun\ and \prob, it remains to complement some of the obtained \XP\ algorithms (those depicted in blue in \cref{fig:hetero_binary_landscape}) with \W\textup{[1]}-hardness results to exclude fixed-parameter tractability.
In this respect, we show that resolving any of the open cases for \prob\ by establishing \W[1]-hardness (which we conjecture to hold)  immediately implies the \W\textup{[1]}-hardness of all the open cases depicted in \cref{fig:hetero_binary_landscape}.
In fact, we show that all~$5$ ``open'' parameterizations for \prob\ are complexity-theoretically equivalent to each other (Theorem~\ref{thm:fpt-equivalent}).
Hence, while settling these $5$ open questions seems to require novel ideas and techniques beyond today's state of the art, resolving any of them settles the complexity of the others.

As our final contribution, we complement the complexity-theoretic landscapes depicted in \cref{fig:hetero_binary_landscape} by considering whether tractability can be achieved by exploiting finer structural properties of the network. Indeed, for graph-theoretic problems, it is typical to also consider structural parameterizations of the input graph such as their \emph{treewidth}~\citep{RobertsonS86}, \emph{treedepth}~\citep{sparsity}, \emph{feedback edge number}~\citep{GanianK21,BredereckHKN22}, or \emph{vertex cover number}~\citep{BodlaenderGP23,Chalopin24}.
It is also common to study such problems when the graph is restricted to being \emph{planar}, i.e., it can be drawn on a plane such that none of its edges cross.
This is particularly important for {\sc Network-Caching}, as planar networks will have much less overlapping cache coverage. 
As our last set of results, we classify the complexity of \hetprobbin, \hetprobun, and \prob\ with respect to the above structural restrictions; in particular, we show that almost none of these restrictions help achieve tractability, even when combined with those considered in \cref{fig:hetero_binary_landscape}.

\section{Setup and Problem Definitions}\label{sec:prelims}
We use standard graph-theoretic terminology~\citep{diestel}.
For any positive integer $n$, let $[n]=\{1,\ldots,n\}$.
We consider a set $\mathcal{\caches}$ of $ \caches := |\mathcal{\caches}|$ caches and a set $\mathcal{\users}$ of $ \users := |\mathcal{\users}|$ users, each of which has access to a subset of caches~(which could be determined by, e.g., routing or other network policies).
This naturally defines a bipartite graph $G=(\mathcal{\caches}, \mathcal{\users}, \mathcal{E})$, where $\mathcal{E}$ are the edges between the two independent sets $\mathcal{\caches}$ and $\mathcal{\users}$, as depicted in \cref{fig:system_model}. This is the predominant setting in the literature~\citep{femto,paria2021texttt,mhaisen2022optimistic,tsigkari2022approximation,salem2023toward}.
For all $u\in\mathcal{\users}$, the subset of caches that $u$ has access to is the neighborhood of $u$, denoted by $N(u)$. 
For all $c\in \caches$, the subset of users that $c$ serves is $N(c)$. The maximum degree of the network is  $\Delta:=\max_{v\in V(G)}|N(v)|$.
The users have access to a catalog $\mathcal{\catalog}$ of $\catalog := |\mathcal{\catalog}|$ contents, where the size of a content $s\in \mathcal{\catalog}$ is denoted by $\size(s)$; as is common in the literature~\citep{femto,blaszczyszyn2015optimal,paschos2019learning} and, without loss of generality, the content sizes are positive integers.
In our analysis, we consider both the cases of contents of  equal (homogeneous) and of variable (heterogeneous) sizes, where the former is equivalent to rescaling and fixing $\size(s)=1$ for all $s\in \mathcal{\catalog}$.

For each cache $c\in \mathcal{\caches}$, we denote by $\capa(c)\in \mathbb{Z^+}$ its capacity, and we denote the maximum capacity of a cache (over all caches) by $\mcap:=\max_{c\in \mathcal{\caches}}\capa(c)$.
As a cache with capacity equal to $\sum_{s\in \mathcal{\catalog}}\size(s)$ can store all the contents in $\mathcal{S}$, without loss of generality, we assume that $\mcap\leq \sum_{s\in \mathcal{\catalog}}\size(s)$, which, in the case of \prob, implies that $\mcap\leq \catalog$.
However, in real systems, $\mcap$ is far from this upper bound.
Indeed, in the case of Netflix, the cache capacity may be as little as 0.3\% of the catalog size~\citep{Paschos-misconceptions,netflix_appliances}. 

\begin{figure}[!t]
	\centering  
	\includegraphics[width=0.96\columnwidth, trim={1cm 7.6cm 4.5cm 4.3cm},clip]{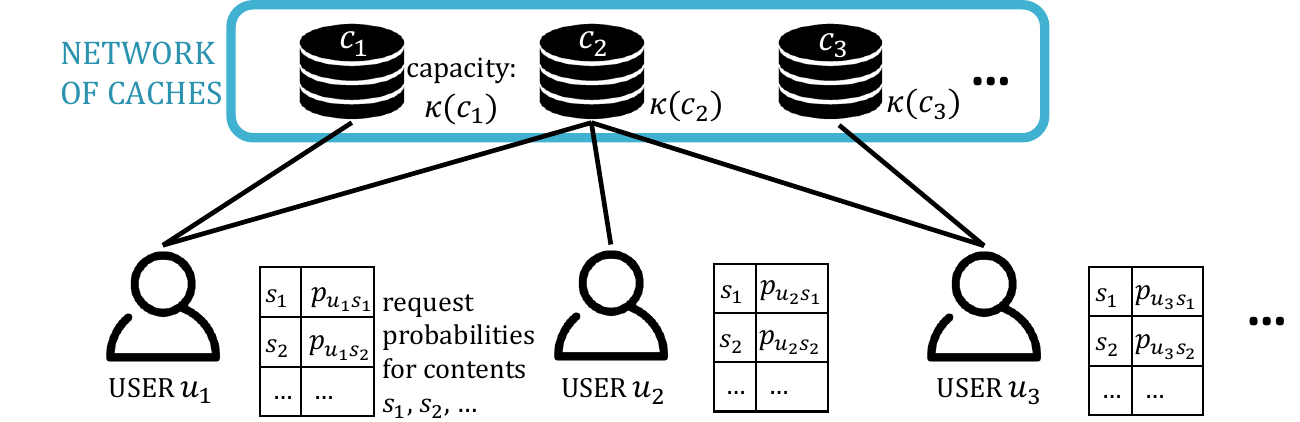} 
	\caption{Illustration of a bipartite graph between users and subsets of the caches in our setting.} 
	\label{fig:system_model}
\end{figure}

A user requests contents following a probability distribution.
Specifically, a user $u\in \mathcal{\users}$ requests content $s\in \mathcal{\catalog}$ with probability $p_{us}\in [0,1]$, such that $\sum_{s\in \mathcal{\catalog}} p_{us} =1$ for all $u\in \mathcal{\users}$. 
The values $p_{us}$ are considered to be known and represented as rational numbers; these may be based on past user requests, trending contents, etc. 
We denote by $\like$ the maximum number of contents any user may request, i.e., $\like := \max_{u\in \mathcal{\users}}  | \{ s\in \mathcal{\catalog}|\; p_{us} >0 \} |$. 
A request for content $s$ by user $u$ is served by a cache adjacent to user $u$ where the requested content is stored.
If the content is not stored in any of the adjacent caches, then the request is served by a large cache outside of $\mathcal{C}$ containing all contents; this will induce traffic at the backhaul network with a potential impact on user quality and retrieval costs~\citep{femto,poularakis2014}.
Thus, it is important to optimize the caching allocation in the ``smaller'' caches in proximity of the users; these are updated during off-peak hours~\citep{netflix_fill_patterns} and the cache placement is static between two updates.
Moreover, each user $u\in \mathcal{\users}$ is characterized by a priority weight $w(u)\in \mathbb{Q^+}$ that may, e.g., capture the guarantee of high quality of service~(QoS) for premium users in streaming platforms~\citep{disney_plus_pricing, netflix_plans_pricing}.

Let a \emph{caching allocation} $Z:\mathcal{\caches} \rightarrow 2^{\mathcal{\catalog}}$ be a mapping that assigns a subset of contents to each cache. If $Z$ complies with the size constraints of the individual caches, i.e., $\forall c\in \mathcal{\caches}: \sum_{s\in Z(c)}\size(s) \leq \capa(c)$, then it is \emph{feasible}.
For a user $u\in \mathcal{\users}$, the set of contents stored in its adjacent small caches (i.e., in $N(u)$) under a caching allocation $Z$ is defined as the set $H(u):=\{s\in \mathcal{\catalog}~|~\exists c\in N(u): s\in Z(c)\}$.
Potential requests for contents in $H(u)$ will induce cache hits. 
We can now define the \emph{cache hit rate}~\citep{femto,paschos2020cache} achieved by a caching allocation $Z$ to be $CH (Z) := \sum_{u\in\mathcal{\users}} \sum_{s\in H(u)} w(u) \cdot p_{us}$. 

As is common in complexity-theoretic studies, we consider the decision variant of the problem in question; we note that all of our algorithms can also solve the corresponding optimization problem and are constructive, i.e., they can output the optimal feasible caching allocation.
We can now define the problem archetype studied in this paper.

      \noindent
      \begin{center}
      \begin{tcolorbox}[title=\textsc{Network-Caching},left=0mm,top=0mm,bottom=0mm,right=0mm,boxsep=1mm]
        \begin{tabular}{@{}ll}
            \textbf{Input:\hspace{-0.3cm}} & \parbox[t]{.78\columnwidth}{A bipartite graph $G=(\mathcal{\caches}, \mathcal{\users}, \mathcal{E})$, the cache capacities $\capa(c_1)\ldots,\capa(c_{\caches})$, a set $\mathcal{\catalog}$ of contents of sizes $\size(s_1),\ldots,\size(s_{\catalog})$, the request probabilities $p_{us}$ for all $u\in \mathcal{\users}$ and $s\in \mathcal{\catalog}$, the user weights $w(u)$ for all $u\in \mathcal{\users}$, and $\score\in \mathbb{Q^+}$.}\\
            \addlinespace
            \textbf{Question:\hspace{-0.3cm}} & \parbox[t]{.78\columnwidth}{Is there a feasible caching allocation $Z$ such that $CH(Z)\geq \score$?}\\
        \end{tabular}
      \end{tcolorbox}
      \end{center}

As the complexity of \textsc{Network-Caching} depends on how the sizes of contents are encoded and represented,  we distinguish between the following three variants of it:

\begin{itemize}
\item In \textsc{Homogeneous Network-Caching} (\textsc{HomNC}), the sizes of the contents are the same, and thus, without loss of generality, are equal to $1$, i.e., $\size(s_1)=\ldots=\size(s_{\catalog})=1$. This restriction has been considered in past works~\citep{blaszczyszyn2015optimal,paschos2019learning}, and can be justified on the basis that contents could in practice be partitioned into equal-sized chunks and cached independently, see \citep{maggi2018adapting}.

\item In \textsc{Heterogeneous Network-Caching (Unary)} (\textsc{HetNC-U}), the numbers $\size(s_1),\ldots,\size(s_{\catalog})$ are encoded in unary. For our complexity-theoretic analysis, this is equivalent to assuming that the sizes of contents are not much larger than the input size. \probtwo\ captures the situation where we do not represent the sizes of contents by their exact bit size, but rather by a rough categorical scale. As an example from the video-on-demand setting, movies are typically roughly twice as long as 1-hour shows, which are then roughly twice as long as 30-minute shows, and hence, we can represent the sizes of these items by the integers $4$, $2$, and $1$, respectively.

\item In \textsc{Heterogeneous Network-Caching (Binary)} (\textsc{HetNC-B}), the numbers $\size(s_1),\ldots,\size(s_{\catalog})$ are encoded in binary, meaning that they can be exponentially larger than the input size. This situation may arise, e.g., if one were to simply use the exact sizes of contents.
\end{itemize}

As mentioned earlier, we study three fundamental variants of \textsc{Network-Caching} with respect to all combinations of the parameterizations discussed in Section~\ref{subsec:contributions}.
As $\caches$, $\mcap$, $\catalog$, and $\users$ are all natural input parameters of this problem, it is only logical to consider parameterizations by them.
On the other hand, we also consider the auxiliary parameters $\Delta$ and $\like$, as one could naturally think that this could lead to tractability.
Indeed, if each user only has access to a bounded number of caches and, additionally, each cache only serves a bounded number of users, then, intuitively, this should render the problem significantly easier.
Further, if each user only requests (with non-zero probability) a bounded number of contents, then, instinctively, this should also simplify the problem.
However, as our first lower bound result (Theorem~\ref{thm:hard-NAE-sat}), we surprisingly show that this is not the case.

\section{Establishing Upper Bounds: Algorithms}\label{sec:upper-bounds}

In this section, we provide all the tractability results delimiting the boundaries of tractability depicted in \cref{fig:hetero_binary_landscape}, i.e., algorithms.
We first analyze \hetprobbin, as any tractability results for this problem also carry over to the other two.
We use a direct branching argument to show:

\begin{theorem}\label{thm:C+S}
\hetprobbin\ is \FPT\ parameterized by $\caches+\catalog$.
\end{theorem}

  \begin{proof}
We apply a brute-force algorithm that computes all the possible caching allocations and finds a feasible one with the largest cache hit rate.
For each cache, it branches over the $2^{\catalog}$ possibilities of storing contents in that cache.
Thus, there are $\OO(2^{\catalog \caches})$ possible caching allocations.
For each caching allocation, it takes $\OO(\caches\cdot \catalog)$ time to test if it is feasible and $\OO(\users\cdot \caches \cdot \catalog)$ time to compute its cache hit rate.
Thus, the algorithm's total runtime is 
$\OO(2^{\catalog \caches}\cdot \caches \cdot \catalog \cdot \users)$. 
Its correctness follows since it enumerates all caching allocations.
\end{proof} 

With Theorem~\ref{thm:C+S} in hand, we can provide a second fixed-parameter tractable fragment for \hetprobbin.

\begin{restatable}{corollary}{corUS}
\label{cor:U+S}
\hetprobbin\ is \FPT\ parameterized by $\users+\catalog$.
\end{restatable}

  \begin{proof}
We apply a simple preprocessing procedure where, for each set of $\catalog+1$ or more caches with the same neighborhood in $G$ (i.e., which are accessible by precisely the same set of users), we keep the $\catalog$ caches with the largest capacities and delete the rest.
We call the resulting graph $G'$. Note that it takes $\OO(\caches \cdot \log(\caches) \cdot \users)$ time to construct $G'$ from $G$.

To see that $(G,\score)$ is a yes-instance if and only if $(G',\score)$ is a yes-instance, for any set of caches with the same neighborhood in $G$, it suffices to note that, as there are only $\catalog$ contents, at most $\catalog$ of these caches can contain a distinct content stored in them that is not stored in any of the other caches in this set.
Thus, for any such set of caches, since only the caches with the smallest capacities are deleted, then any set of contents that can be stored prior to deleting these caches, can also be stored after deleting these caches.

Since there are $2^{\users}$ different subsets of $\mathcal{\users}$, there are at most $2^{\users}\cdot \catalog$ caches in $G'$.
At this point, it suffices to apply the brute force algorithm from Theorem~$1$ on $G'$.
Thus, the total runtime of this algorithm (including constructing $G'$ from $G$) is $\OO(\caches \cdot \log(\caches) \cdot \users + 2^{\catalog \cdot (2^{\users}\cdot \catalog)}\cdot 2^{\users}\cdot \catalog \cdot \catalog \cdot \users)=\OO(\caches \cdot \log(\caches) \cdot \users + 2^{\catalog^2 \cdot 2^{\users}+\users}\cdot \catalog^2 \cdot \users)$, and the correctness follows immediately from the argument provided in the previous paragraph.
\end{proof} 

We now have all the tractability results needed for \cref{fig:hetero_binary_landscape}~(top).
For \hetprobun, we contrast the situation for the binary case by designing a dynamic programming algorithm whose running time is \XP\ when parameterized by $\caches$ alone, and even \FPT\ when parameterized by $\caches+\mcap$.

\begin{restatable}{theorem}{thmxpc}
\label{thm:xp-c}
\hetprobun\ is \XP\ parameterized by $\caches$ and \FPT\ parameterized by $\caches+\mcap$. 
\end{restatable}

 \begin{proof}
Let $c_1,\ldots,c_{\caches}$ be the caches and $s_1,\ldots,s_{\catalog}$ the contents of the catalog.
Let $A$ be a $\caches$-dimensional array where, for all $i\in [\caches]$, the $i^{\text{th}}$ dimension of the array has size $\capa(c_i)+1$.
We use a zero-based indexing for $A$, that is, the indices start from $0$.
Essentially, each entry in $A$ will correspond to the best possible cache hit rate that can be achieved when the caches have the remaining capacities corresponding to the coordinates of the entry, with the value of the $i^{\text{th}}$ coordinate corresponding to the remaining capacity of the cache $c_i$ for all $i\in [\caches]$.
Initially, we set all of the entries of $A$ to $-1$ except for one entry that we set to $0$, whose coordinates correspond to all of the caches having their full capacities remaining (i.e., $A[\capa(c_1)]\ldots[\capa(c_{\caches})]=0$).
The idea here is that $-1$ is a placeholder that represents that a caching allocation resulting in the represented cache capacities cannot yet exist, which is indeed the case here since the capacities of the caches cannot be reduced without storing contents in them.
Then, we will update the entries of $A$ through a dynamic programming approach considering the contents one by one and all the possibilities of adding the content to the caches.
In order to perform this update properly, we will also have a second array $A'$ with $A'=A$ initially.
The algorithm proceeds as follows.

For $i=1$ to $i=\catalog$, do the following.
For each non-negative entry $A[j_1]\ldots[j_{\caches}]$ in $A$ and each of the $2^{\caches}$ subsets $\mathcal{\caches}'\subseteq \mathcal{\caches}$, possibly update the entries of $A'$ as follows.
Let $CH(\mathcal{\caches}',s_i)$ be the cache hits obtained only from storing $s_i$ in the caches of $\mathcal{\caches}'$.
For all $t\in [\caches]$, let $q_{t}=1$ if $c_t\in \mathcal{\caches}'$, and otherwise, let $q_t=0$.
If, for all $t\in [\caches]$, it holds that $j_t-(q_t \cdot \size(s_i))\geq 0$, and 
$$A[j_1] \ldots [j_{\caches}] + CH(\mathcal{\caches}',s_i) > $$ $$A'[j_1-(q_1 \cdot \size(s_i))] \ldots [j_{\caches}-(q_{\caches} \cdot \size(s_i))],$$
then set 
$$A'[j_1-(q_1 \cdot \size(s_i))] \ldots [j_{\caches}-(q_{\caches} \cdot \size(s_i))] =$$ $$A[j_1] \ldots [j_{\caches}] + CH(\mathcal{\caches}',s_i).$$
This ends the two innermost For loops.
Before incrementing $i$ by $1$, set $A=A'$.

Once the above algorithm is finished, let $\score'$ be the largest entry in $A'$. Note that $\score'$ corresponds to the maximum cache hit rate possible among all feasible caching allocations.
Hence, if $\score'\geq \score$, then return that this is a yes-instance, and otherwise, return that this is a no-instance.

Now, we prove the correctness of the algorithm by induction on the number $i$ of contents processed.
For the base case, $i=0$ and all the entries in $A'$ are $-1$ except for $A'[\capa(c_1)]\ldots,[\capa(c_{\caches})]$ which is $0$, and this is correct since it is not possible to have reduced cache capacities nor to have a positive cache hit rate without storing any contents in caches.
Trivially, all of the possible ways (in fact, the only way) of storing zero contents have been considered.
For the inductive hypothesis, suppose for some $0\leq i < \catalog$ that after the $i^{\text{th}}$ iteration of the outermost For loop has completed (before the first iteration of the outermost For loop has begun in the case $i=0$), each entry in $A'$ is correct and every possible feasible caching allocation for the first $i$ contents have been considered thus far.
Specifically, in this case we say that an entry in $A'$ is correct if (1) it is $-1$ if it is not possible to achieve the remaining cache capacities by only storing contents from the first $i$ contents or (2) it is the maximum possible value of the cache hit rate that can be achieved by a feasible caching allocation for the first $i$ contents such that the remaining capacities of the caches correspond to the coordinates of the entry in $A'$.

We now prove the inductive step for $i+1$.
In the $(i+1)^{\text{th}}$ iteration of the outermost For loop of the algorithm, initially $A=A'$.
Then, for each non-negative entry $A[j_1]\ldots[j_{\caches}]$ in $A$, all of the possible ways of caching the content $s_{i+1}$ in the caches whose remaining capacities correspond to the coordinates of $A[j_1]\ldots[j_{\caches}]$ are considered.
This covers all the possible feasible caching allocations for the first $i+1$ contents since, by the inductive hypothesis, all of the entries in $A$ are correct, and thus, the non-negative entries in $A$ should not be considered as they correspond to remaining cache capacities that are unattainable by storing contents from the first $i$ contents.
Now, the only entries in $A'$ that can be updated from considering $A[j_1]\ldots[j_{\caches}]$ are correctly those whose coordinates are attainable from $A[j_1]\ldots[j_{\caches}]$ by storing the content $s_{i+1}$ in a subset of the caches for all such subsets that do not result in negative coordinates.
For a subset $\mathcal{\caches}'\subseteq \mathcal{\caches}$ of the caches, the corresponding entry in $A'$ is only updated if the current value of that entry in $A'$ is less than $A[j_1]\ldots[j_{\caches}]+CH(\mathcal{\caches}',s_{i+1})$.
We now argue that this update step is correct.
First, if the above condition is not met, then the entry in $A'$ should not be updated.
Indeed, by the inductive hypothesis, there exists a feasible caching allocation of the first $i$ contents that achieves at least as large a cache hit rate as the one that would be obtained from $A[j_1]\ldots[j_{\caches}]$ and $\mathcal{\caches}'$, and they both have the same remaining cache capacities.
Second, if the above condition is met, then the entry in $A'$ should be updated.
Indeed, by the inductive hypothesis, there exists a feasible caching allocation for the first $i$ contents whose cache hit rate is $A[j_1]\ldots[j_{\caches}]$ and whose remaining cache capacities for the caches $c_1,\ldots,c_{\caches}$ are $j_1,\ldots,j_{\caches}$, respectively.
Thus, there also exists a feasible caching allocation for which the remaining cache capacities correspond to the coordinates of the entry of $A'$ being updated, and whose cache hit rate is $A[j_1]\ldots[j_{\caches}]+CH(\mathcal{\caches}',s_{i+1})$ by the definition of $CH(\mathcal{\caches}',s_{i+1})$.
Hence, all of the entries of $A'$ are correct after the $(i+1)^{\text{th}}$ iteration of the outermost For loop has completed, and thus, we have proven the inductive step.
It is then clear that $\score'$ corresponds to the maximum cache hit rate possible among all feasible caching allocations, and so, the algorithm decides correctly on whether or not it is a yes- or no-instance of the problem.

Finally, we prove the runtime of the algorithm.
For each $\mathcal{\caches}'\subseteq \mathcal{\caches}$ and $i\in [\catalog]$, it takes $\OO(\caches\cdot \users)$ time to compute $CH(\mathcal{\caches}',s_i)$ and to check whether, for all $t\in [\caches]$, $j_t-(q_t \cdot \size(s_i))\geq 0$.
Since the largest cache hit rate of any of the entries in $A'$ can be stored and updated in the last iteration of the outermost For loop, $\score'$ can be extracted from $A'$ in $\OO(1)$ time.
Thus, for any $i\in [\catalog]$, the total runtime for the $i^{\text{th}}$ iteration of the outermost For loop of the algorithm takes $\OO((\mcap+1)^{\caches} \cdot 2^{\caches}\cdot \caches \cdot \users)$ time.
As there are $\catalog$ such iterations of this For loop, the total runtime of the algorithm is $\OO((\mcap+1)^{\caches} \cdot 2^{\caches}\cdot \caches \cdot \users \cdot \catalog)$.  
\end{proof} 

The ideas from the above algorithm can be extended to show that \hetprobun\ is \XP\ parameterized by $\users+\mcap$. This is done by grouping caches together into types based on the users they serve and their capacities, and keeping track of the number of caches of each type that are left as we fill them.

\begin{restatable}{theorem}{thmxpud}
\label{thm:xp-ud}
\hetprobun\ is \XP\ parameterized by $\users+\mcap$.
\end{restatable}

 \begin{proof}
Let $s_1,\ldots,s_{\catalog}$ be the contents of the catalog.
Also, let $\mathcal{\users}_1,\ldots,\mathcal{\users}_{2^{\users}}$ be the distinct subsets of $\mathcal{\users}$ and, for all $r\in [2^{\users}]$, let $\mathcal{\caches}_r:=\{c \mid N(c)=\mathcal{\users}_r~\text{and}~c\in \mathcal{\caches}\}$.
Further, for all integers $0\leq y\leq [\mcap]$, let $\mathcal{\caches}_r^y$ be the set of caches in $\mathcal{\caches}_r$ whose remaining capacities (after possibly storing some contents) are equal to $y$.
Thus, there are $T:=2^{\users}\cdot (\mcap+1)$ sets of the form $\mathcal{\caches}_r^y$ or, in other words, types of caches.
Let $A$ be a $T$-dimensional array where, for all $i\in [T]$, the $i^{\text{th}}$ dimension of the array corresponds to a cache type and has size $\caches+1$.
Specifically, for all $r\in [2^{\users}]$ and integers $0\leq y \leq \mcap$, the cache of type $\mathcal{\caches}_r^y$ corresponds to the $((r-1) \cdot (\mcap+1) + y + 1))^{\text{th}}$ dimension of $A$, and this dimension's coordinate corresponds to the number of caches of type $\mathcal{\caches}_r^y$ remaining.
We use a zero-based indexing for $A$, that is, the indices start from $0$.
Essentially, each entry in $A$ will correspond to the best possible cache hit rate that can be achieved when the numbers of each type of cache remaining correspond to the coordinates of the entry.
Initially, we set all of the entries of $A$ to $-1$ except for one entry that we set to $0$, whose coordinates correspond to the initial numbers of each type of cache before any contents are stored in caches.
The $-1$ is a placeholder and plays an analogous role as in the proof of Theorem~\ref{thm:xp-c}.
Then, we will update the entries of $A$ through a dynamic programming approach considering the contents one by one and all the possibilities of adding the content to types of caches.
In order to do this update properly, we will also have a second array $A'$ with $A'=A$ initially.
The algorithm proceeds as follows.

For $i=1$ to $i=\catalog$, do the following.
For each non-negative entry $A[j_1]\ldots[j_T]$ in $A$ and each of the $2^T$ subsets of $\mathcal{\caches}'\subseteq \{\mathcal{\caches}_r^y \mid r\in [2^{\users}]~\text{and}~0\leq y\leq \mcap\}$, possibly update the entries of $A'$ as follows.
Let $CH(\mathcal{\caches}',s_i)$ be the cache hits obtained only from storing $s_i$ in the types of caches of $\mathcal{\caches}'$.
When storing $s_i$ in the types of caches of $\mathcal{\caches}'$, this changes the number of certain types of caches remaining, possibly increasing or decreasing it for some.
To simplify matters, without loss of generality, from $j_1,\ldots,j_T$, after storing the content $s_i$ in the types of caches of $\mathcal{\caches}'$, let $j'_1,\ldots,j'_{T}$ be the resulting numbers of each cache type remaining.
If, for all $t\in [T]$, it holds that $j'_t\geq 0$, and 
$$A[j_1] \ldots [j_T] + CH(\mathcal{\caches}',s_i) > A'[j'_1] \ldots [j'_T],$$
then set 
$$A'[j'_1] \ldots [j'_T] =  A[j_1] \ldots [j_T] + CH(\mathcal{\caches}',s_i).$$
This ends the two innermost For loops.
Before incrementing $i$ by $1$, set $A=A'$.

Once the above algorithm is finished, let $\score'$ be the largest entry in $A'$ and note that $\score'$ corresponds to the maximum cache hit rate possible among all feasible caching allocations.
Hence, if $\score'\geq \score$, then return that this is a yes-instance, and otherwise, return that this is a no-instance.

Now, we prove the correctness of the algorithm by induction on the number $i$ of contents processed.
For the base case, $i=0$ and all the entries in $A'$ are $-1$ except for one which is $0$, whose coordinates correspond to the initial numbers of each type of cache before any contents are stored in caches, and this is correct since it is not possible to have reduced cache capacities nor to have a positive cache hit rate without storing any contents in caches.
Trivially, all of the possible ways (in fact, the only way) of storing zero contents have been considered.
For the inductive hypothesis, suppose for some $0\leq i < \catalog$ that after the $i^{\text{th}}$ iteration of the outermost For loop has completed (before the first iteration of the outermost For loop has begun in the case $i=0$), each entry in $A'$ is correct and every possible feasible caching allocation for the first $i$ contents have been considered thus far.
Specifically, in this case we say that an entry in $A'$ is correct if (1) it is $-1$ if it is not possible to achieve the remaining numbers of each cache type by only storing contents from the first $i$ contents or (2) it is the maximum possible value of the cache hit rate that can be achieved by a feasible caching allocation for the first $i$ contents such that the remaining numbers of each cache type correspond to the coordinates of the entry in $A'$.

We now prove the inductive step for $i+1$.
In the $(i+1)^{\text{th}}$ iteration of the outermost For loop of the algorithm, initially $A=A'$.
Then, for each non-negative entry $A[j_1]\ldots[j_T]$ in $A$, all of the possible ways of caching the content $s_{i+1}$ in types of caches whose remaining numbers correspond to the coordinates of $A[j_1]\ldots[j_T]$ are considered.
This covers all the possible feasible caching allocations for the first $i+1$ contents since, by the inductive hypothesis, all of the entries in $A$ are correct, and thus, the non-negative entries in $A$ should not be considered as they correspond to remaining numbers of cache types that are unattainable by storing contents from the first $i$ contents.
Also, it should be noted that the previous statement is true since there is no sense in storing the same content in two or more caches of the same type as they serve the same subset of the users.
Now, the only entries in $A'$ that can be updated from considering $A[j_1]\ldots[j_T]$ are correctly those whose coordinates are attainable from $A[j_1]\ldots[j_T]$ by storing the content $s_{i+1}$ in a subset of the cache types for all such subsets that do not result in negative coordinates.
For a subset $\mathcal{\caches}'\subseteq \{\mathcal{\caches}_r^y \mid r\in [2^{\users}]~\text{and}~0\leq y\leq \mcap\}$ of the cache types, the corresponding entry in $A'$ is only updated if the current value of that entry in $A'$ is less than $A[j_1]\ldots[j_T]+CH(\mathcal{\caches}',s_{i+1})$.
We now argue that this update step is correct.
First, if the above condition is not met, then the entry in $A'$ should not be updated.
Indeed, by the inductive hypothesis, there exists a feasible caching allocation of the first $i$ contents that achieves at least as large a cache hit rate as the one that would be obtained from $A[j_1]\ldots[j_T]$ and $\mathcal{\caches}'$, and they both have the same remaining numbers of cache types.
Second, if the above condition is met, then the entry in $A'$ should be updated.
Indeed, by the inductive hypothesis, there exists a feasible caching allocation for the first $i$ contents whose cache hit rate is $A[j_1]\ldots[j_T]$ and whose remaining numbers of cache types are $j_1,\ldots,j_T$, in that order.
Thus, there also exists a feasible caching allocation for which the remaining numbers of cache types correspond to the coordinates of the entry of $A'$ being updated, and whose cache hit rate is $A[j_1]\ldots[j_T]+CH(\mathcal{\caches}',s_{i+1})$ by the definition of $CH(\mathcal{\caches}',s_{i+1})$.
Hence, all of the entries of $A'$ are correct after the $(i+1)^{\text{th}}$ iteration of the outermost For loop has completed, and thus, we have proven the inductive step.
It is then clear that $\score'$ corresponds to the maximum cache hit rate possible among all feasible caching allocations, and so, the algorithm decides correctly on whether or not it is a yes- or no-instance of the problem.

Finally, we prove the runtime of the algorithm.
As in the proof of Theorem~\ref{thm:xp-c}, for each $\mathcal{\caches}'\subseteq \{\mathcal{\caches}_r^y \mid r\in [2^{\users}]~\text{and}~0\leq y\leq \mcap\}$ and $i\in [\catalog]$, it takes $\OO(\caches\cdot \users)$ time to compute $CH(\mathcal{\caches}',s_i)$.
It takes $\OO(T)$ time to check whether $j'_t\geq 0$ for each $t\in [T]$.
Also as in the proof of Theorem~\ref{thm:xp-c}, at the end $\score'$ can be extracted from $A'$ in $\OO(1)$ time.
Thus, for any $i\in [\catalog]$, the total runtime for the $i^{\text{th}}$ iteration of the outermost For loop of the algorithm takes $\OO((\caches+1)^{T} \cdot 2^T \cdot (\caches \cdot \users + T))$ time.
As there are $\catalog$ such iterations of this For loop, the total runtime of the algorithm is $\OO((\caches+1)^{T} \cdot 2^T \cdot (\caches \cdot \users + T) \cdot \catalog)=\OO((\caches+1)^{f(\users,\mcap)}\cdot \catalog)$ for some computable function $f$.  
\end{proof}

We now have all the tractability results we need to establish \cref{fig:hetero_binary_landscape}~(middle):
the upper bounds follow from Theorems~\ref{thm:C+S},~\ref{thm:xp-c},~and~\ref{thm:xp-ud} plus Corollary~\ref{cor:U+S}.
Finally, we establish that \prob\ is \XP\ parameterized by $\users$ through the following observation that allows us to bound the number of caches by a function of the number of users, after which the \XP\ algorithm provided in Theorem~\ref{thm:xp-c} can be applied.
The reason why this only works for \prob\ is that it amalgamates caches with the same neighborhood into a single cache with a larger capacity, which cannot be safely done (i.e., without changing the outcome of the problem) in the other settings due to the variable content sizes.

\begin{restatable}{observation}{obsusers}
\label{obs:users}
Any instance $(G,\score)$ of \prob\ with $\users$ users can be reduced in polynomial time to an equivalent instance $(G',\score)$ of \prob\ with at most $2^{\users}$ caches.
\end{restatable}

\begin{proof}
To obtain $G'$ from $G$, for each set of $2$ or more caches with the same neighborhood in $G$, delete all the caches except for one, and give this cache a capacity equal to the sum of all the capacities of these deleted caches plus its own original capacity.
Since there are $2^{\users}$ different subsets of the users in $G$, there are at most $2^{\users}$ caches in $G'$.
This reduction takes polynomial time and the equivalence of the two instances is immediate as all the contents have unit size.
\end{proof} 

Theorem~\ref{thm:xp-c} together with Observation~\ref{obs:users} immediately yields  the following corollary.

\begin{corollary}\label{cor:xp-u}
\prob\ is \XP\ parameterized by $\users$.
\end{corollary}

Corollary~\ref{cor:xp-u} is the only additional tractability result needed for the landscape in \cref{fig:hetero_binary_landscape}~(bottom).

\section{Establishing Lower Bounds: Hardness}\label{sec:lower-bounds}

In this section, we establish the lower bounds required for the landscapes depicted in \cref{fig:hetero_binary_landscape}. 
As we are proving hardness results, it is advantageous to first consider parameterizations of \prob\ as they will carry over to \hetprobun\ and \hetprobbin. 
To simplify our exposition, we denote all of the instances constructed in our reductions as a bipartite graph~$G$ together with the target cache hit rate~$\score$.

We first strengthen the known result~\citep{femto} that \prob\ is \NP-hard even if $\catalog=\like=2$ and $\mcap=1$, by showing that the same holds even if additionally restricted to networks where $\Delta=3$.
We prove this via a reduction from the \NP-hard {\sc Monotone NAE-3-SAT-B3} problem~\citep{KT02}, whose definition is as follows.

\defproblem{{\sc Monotone NAE-3-SAT-B3}}
{A 3-CNF formula $\phi$ in which each clause contains $2$ or $3$ literals, and every variable appears in at most $3$ clauses and only in its positive form.}
{Is there an {\sc NAE} satisfying assignment for $\phi$, i.e., a truth assignment to the variables appearing in $\phi$ such that each clause in $\phi$ contains both a variable set to True and a variable set to False?}

\begin{restatable}{theorem}{thmparameters}
\label{thm:hard-NAE-sat}
\prob\ is \NP-hard, even if $\catalog=\like=2$, $\mcap=1$, and $\Delta=3$.
\end{restatable}

\begin{proof}
We reduce from {\sc Monotone NAE-3-SAT-B3}.
Let $x_1,\ldots,x_n$ and $C_1,\ldots,C_m$ be the variables and clauses, respectively, of an input formula $\phi$ of {\sc Monotone NAE-3-SAT-B3}.
We construct an instance $(G,\score)$ of \prob\ from the incidence graph\footnote{The incidence graph, denoted by $G_{\phi}$, is the bipartite graph obtained by creating a vertex for each variable and clause in $\phi$, and adding an edge between a clause vertex and a variable vertex if and only if that variable is contained in that clause in $\phi$.} $G_{\phi}$ of $\phi$ as follows.
Each variable vertex $x_i$, $i\in [n]$, in $G_{\phi}$ corresponds to a cache $c_{x_i}$ of capacity $1$ in $G$.
There are only two contents in the catalog: True and False.
Each clause vertex $C_j$, $j\in [m]$, in $G_{\phi}$ corresponds to a user $u_{C_j}$ (of weight $1$) in $G$ that requests True and False with equal probability.
Thus, if a variable appears in a clause in $\phi$, then the variable's corresponding cache is adjacent to the clause's corresponding user in $G$.
Lastly, set $\score:=\users$.
This completes the construction of $(G,\score)$, which is clearly achieved in polynomial time.
Further, it is easy to verify that, in $(G,\score)$, $\catalog=2$, $\mcap=1$, $\like=2$, and $\Delta=3$ (as each variable appears in at most $3$ clauses in $\phi$).

For the first direction, assume that there is an {\sc NAE} satisfying assignment for $\phi$.
For each $i\in [n]$, if the {\sc NAE} satisfying assignment for $\phi$ sets the variable $x_i$ to True (False, resp.), then store the content True (False, resp.) in the cache $c_{x_i}$ in $G$.
The {\sc NAE} satisfying assignment for $\phi$ ensures that each clause in $\phi$ contains a variable set to True and one set to False.
Thus, this is a feasible caching allocation in $G$ that ensures that each user is adjacent to a cache that stored True and one that stored False, and so, it achieves a cache hit rate of $\score$.
Indeed, each clause in $\phi$ corresponds to a user in $G$ that is adjacent to $2$ or $3$ caches that correspond to the $2$ or $3$ variables contained in that clause in $\phi$.

For the other direction, assume that there is a feasible caching allocation in $G$ achieving a cache hit rate of $\score$.
This caching allocation ensures that each user $u_{C_j}$, $j\in [m]$, is adjacent to a cache that stored True and one that stored False.
For each variable $x_i$, $i\in [n]$, in $\phi$, set its truth value to the one that corresponds to the content stored in $x_i$'s corresponding cache $c_{x_i}$ in $G$ according to this caching allocation.
Then, this truth assignment is an {\sc NAE} satisfying assignment for $\phi$ since each clause contains a variable set to True and one set to False.
Indeed, each user in $G$ corresponds to a clause in $\phi$ that contains $2$ or $3$ variables corresponding to the $2$ or $3$ caches adjacent to that user in $G$.
\end{proof}

Theorem~\ref{thm:hard-NAE-sat} settles the complexity of our three problems with respect to all our considered parameters except for $\caches$ and $\users$. For these two parameters, the complexity differs depending on the problem and, as we showed in Section~\ref{sec:upper-bounds}, \prob\ is \XP\ parameterized by $\caches$ or $\users$. 
For \hetprobun, the three additional lower bound results we need to establish the landscape in \cref{fig:hetero_binary_landscape}~(middle) are \W\textup{[1]}-hardness with respect to $\caches+\lambda$ and $\caches+\users$, and \paraNP-hardness when parameterized by $\users$ alone. For the latter two lower bounds, we provide a simple reduction from \textsc{Unary Bin Packing}, which is not only \NP-hard, but also \W\textup{[1]}-hard when parameterized by the number of bins~\citep{Jansen2013}. Its definition is as follows.

\defproblem{{\sc Unary Bin Packing}}
{A set $I$ of item sizes that are positive integers encoded in unary, and $b,B\in \mathbb{Z}^+$.}
{Is there a partition of the items in $I$ into $b$ bins of capacity $B$?}

\begin{restatable}{theorem}{thmoneuser}
\label{thm:one-user}
\hetprobun\ is \NP-hard and also \W\textup{[1]}-hard parameterized by $\caches$, even if $\users=1$.
\end{restatable}

 \begin{proof}
We reduce from {\sc Unary Bin Packing}.
From an instance $(I,b,B)$ of {\sc Unary Bin Packing}, we construct an instance $(G,\score)$ of \hetprobun\ as follows.
There is a single user $u$ (of weight $1$) and $b$ caches of capacity $B$ that are all adjacent to $u$.
For each item size in $I$, there is a content of the same size and the user $u$ requests each of the $|I|$ contents with equal probability.
Set $\score:=1$.
This completes the construction of $(G,\score)$, which is clearly achieved in polynomial time.
Further, it is easy to verify that, in $(G,\score)$, $\caches=b$ and $\users=1$.

To establish correctness, it suffices to observe that our construction maintains a direct correspondence between storing a content in a cache (in the constructed instance of \hetprobun) and placing the corresponding item in the respective bin (in the original instance of {\sc Unary Bin Packing}). Hence, a caching allocation achieves a cache hit rate of $\ell=1$ if and only if the corresponding placement of items in bins represents a solution to the original {\sc Unary Bin Packing} problem.
\end{proof} 

For the last lower bound in the unary case, a direct adaptation of the reduction used in Theorem~\ref{thm:one-user}
yields that \hetprobun\ is \W\textup{[1]}-hard parameterized by $\caches$ when each user requests a single content.

\begin{restatable}{corollary}{corsplitusers}
\label{cor:split-users}
\hetprobun\ is \NP-hard and also \W\textup{[1]}-hard parameterized by $\caches$, even if $\like=1$.
\end{restatable} 

 \begin{proof}
Let $s_1,\ldots,s_{|I|}$ be the contents requested with non-zero probability by the single user $u$ in the reduction in the proof of Theorem~\ref{thm:one-user} (recall that $|I|$ is the number of items in the input instance of \textsc{Unary Bin Packing}).
We adapt that reduction as follows: instead of creating the single user $u$, for each $s\in s_1,\ldots,s_{|I|}$, we create a separate user $u_s$ of weight $1$. Each such user $u_s$ only requests the content $s$, specifically with probability $p_{u_ss}=1$, and is moreover adjacent to every cache in the instance. Finally, we set $\score:=|I|$ instead of $1$.

To establish correctness, we again observe that our construction maintains a direct correspondence between storing a content in a cache and placing the corresponding item in the respective bin. Hence, a caching allocation achieves a cache hit rate of $\ell=|I|$ if and only if the corresponding placement of items in bins represents a solution to the original {\sc Unary Bin Packing} problem.
\end{proof} 

When the sizes of the contents are encoded in binary, one can establish a much stronger notion of intractability (\paraNP-hardness) for \probtwo\ when parameterized by $\caches$. 
Indeed, as was observed for other variants of {\sc Network-Caching}~\citep{PIAKT18}, \hetprobbin\ admits a trivial polynomial-time reduction from the weakly \NP-hard {\sc $0$-$1$ Knapsack} problem~\citep{GareyJ79}, whose definition is as follows. 

\defproblem{{\sc $0$-$1$ Knapsack}}
{A knapsack of capacity $W\in \mathbb{Z}^+$, a set $I$ of items numbered from $1$ to $n$ with weights $w_i\in \mathbb{Z}^+$ and values $v_i\in \mathbb{Z}^+$ ($i\in [n]$), and $t\in \mathbb{Z}^+$, all encoded in binary.}
{Is there a subset of the items that can be placed in the knapsack (i.e., their total weight does not exceed $W$) such that their total value is at least $t$?}

\begin{restatable}{theorem}{knapsack}
\label{thm:knapsack}
\hetprobbin\ is \NP-hard, even if $\users=\caches=1$.
\end{restatable}

\begin{proof}
From an instance $(W,I,t)$ of {\sc $0$-$1$ Knapsack}, we construct an instance $(G,\score)$ of \hetprobbin\ as follows.
For each $i\in [n]$, there is a content $s_i$ of size $\size(s_i):=w_i$ representing the item $i$ of weight $w_i$ from the {\sc $0$-$1$ Knapsack} instance. Further, there is one cache of capacity $W$ representing the knapsack and there is one user $u$ (of weight $1$) adjacent to this cache such that $p_{us_i}:=v_i/\sum\limits_{i=1}^{n} v_i$ for all $i\in [n]$. Set $\score:=t/\sum\limits_{i=1}^{n} v_i$.
This completes the construction of $(G,\score)$, which is clearly achieved in polynomial time.

By observing that the values of the items and $t$ have been normalized so that their corresponding request probabilities and $\score$, respectively, are in the range $[0,1]$, the equivalence of the two instances is immediate.
Indeed, storing a content in the cache corresponds to placing its corresponding item in the knapsack as the size and request probabilities of the contents correspond to the weights and (relative) values of the items, respectively, the cache and the knapsack have the same capacity, and thus, a cache hit rate of at least $\score$ is achieved if and only if there is a subset of the items of total value at least $t$ that can be placed in the knapsack.
\end{proof} 

A straightforward corollary of the proof of Theorem~\ref{thm:knapsack} yields that \hetprobbin\ is \NP-hard, even if there is only one cache and each user requests a single content.

\begin{restatable}{corollary}{corsplituk}
\label{cor:split-users-knapsack}
\hetprobbin\ is \NP-hard, even if $\caches=\like=1$.
\end{restatable}

 \begin{proof}
Let $s_1,\ldots,s_{|I|}$ be the contents requested with probabilities $p_{us_1},\ldots,p_{us_{|I|}}$, respectively, by the single user $u$ in the reduction in the proof of Theorem~\ref{thm:knapsack} (recall that $|I|$ is the number of items in the input instance of {\sc $0$-$1$ Knapsack}).
We adapt that reduction as follows: instead of creating the user $u$, for each $s\in s_1,\ldots,s_{|I|}$, we create a separate user $u_s$ of weight $p_{us}$. Each such user $u_s$ only requests the content $s$, specifically with probability $p_{u_ss}=1$, and is adjacent to the only cache in the instance. We set $\score:=t/\sum\limits_{i=1}^{n} v_i$ as above.

To establish correctness, we again observe that our construction maintains a direct correspondence between storing a content in a cache and placing the corresponding item in the knapsack. Hence, a caching allocation achieves a cache hit rate of $\score:=t/\sum\limits_{i=1}^{n} v_i$ if and only if a subset of the items of total value at least $t$ can be placed in the knapsack.
\end{proof}

For some parameterizations, our results do not resolve whether they are \FPT\ or \W\textup{[1]}-hard. In the final technical Section~\ref{sec:interreduce}, we show that these open cases are interreducible.

\section{Structural Parameters}\label{sec:structural}
In this section, we discuss an alternative approach towards identifying tractable fragments of \prob, specifically by exploiting well-established structural properties of the network.
As we have seen, \hetprobbin\ is \NP-hard even if the network consists of a single edge, and \hetprobun\ is \NP-hard even in star networks.
Thus, they both remain \paraNP-hard when parameterized by not only the fundamental structural parameter treewidth~\citep{RobertsonS86}, but also by essentially all other established graph parameters including, e.g., treedepth~\citep{sparsity}, the feedback edge number~\citep{GanianK21,BredereckHKN22}, and the vertex cover number~\citep{BodlaenderGP23,Chalopin24}. Moreover, this implies that both problems remain \NP-hard on planar networks---a property which is particularly relevant in the studied setting (see also Section~\ref{subsec:contributions}).
However, none of the hardness results presented thus far rule out tractability for \prob\ with respect to these structural parameters or planarity.
As our next result, we show that most established structural parameters like treewidth, treedepth, and feedback edge number, as well as planarity, do not help even when dealing with the simpler homogeneous setting. We achieve this via a reduction from the following problem.

\defproblem{{\sc Maximum $k$-Vertex Cover}}
{A graph $G$ and two positive integers $k$ and $t$.}
{Is there a subset of vertices $V'\subseteq V(G)$ such that $|V'|\leq k$ and at least $t$ edges in $G$ contain a vertex from $V'$?}

\begin{restatable}{theorem}{thmvcred}
\label{thm:vc-red}
\prob\ is 
\NP-hard even if $\like=2$ and $G$ is a star whose edges have each been subdivided once. Moreover, in this case it is also \W\textup{[1]}-hard parameterized by $\mcap$.
\end{restatable}

\begin{figure}
	\centering  
\includegraphics[scale=0.6]{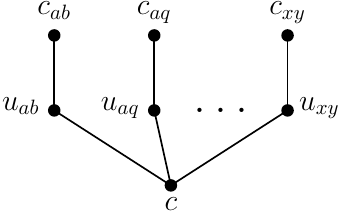}
\caption{Illustration of the subdivided star $G'$ constructed in the proof of Theorem~\ref{thm:vc-red}, where $ab$, $aq$, and $xy$ are edges in the graph $G$ from the instance of {\sc Maximum $k$-Vertex Cover}.}\label{fig:subdivided_star}
\end{figure}

  \begin{proof}
We reduce from {\sc Maximum $k$-Vertex Cover}.
From an instance $(G,k,t)$ of {\sc Maximum $k$-Vertex Cover}, we construct an instance $(G',\score)$ of \prob\ as follows.
For each edge $xy\in E(G)$, there is a user $u_{xy}$ (of weight $1$) in $G'$ that only requests (with non-zero probability) the contents $x$ and $y$ with equal probability (i.e., $p_{u_{xy}x}=p_{u_{xy}y}=0.5$), and a cache $c_{xy}$ of capacity $1$ that is adjacent only to $u_{xy}$.
There is also one cache $c$ of capacity $k$ that is adjacent to every user in $G'$.
Lastly, set $\score:=(\users+t)/2$.
This completes the construction of $(G',\score)$, which is clearly achieved in polynomial time.
Further, it is easy to verify that, in $(G',\score)$, $\like=2$, $\mcap=k$, and $G'$ is a star whose edges have each been subdivided exactly once.
See Figure~\ref{fig:subdivided_star} for an illustration of $G'$.

For the first direction, assume that there is a subset of vertices $V'\subseteq V(G)$ such that $|V'|\leq k$ and at least $t$ edges in $G$ contain a vertex from $V'$.
For each vertex $v\in V'$, store the content $v$ in the cache $c$ in $G'$.
For each edge $xy\in E(G)$ that contains exactly one vertex from $V'$, say $x$, store the content $y$ in the cache $c_{xy}$.
For the rest of the caches of the form $c_{xy}$, cache either the content $x$ or the content $y$.
Then, this is a feasible caching allocation in $G'$ that ensures that at least $t$ users have both of their contents (that they requested with non-zero probability) stored in adjacent caches, and the remaining at most $\users-t$ users have exactly one of their contents (that they requested with non-zero probability) stored in an adjacent cache.
Thus, this feasible caching allocation in $G'$ achieves a cache hit rate of $t+(\users-t)/2=(\users+t)/2=\score$.

For the other direction, assume that their is a feasible caching allocation in $G'$ achieving a cache hit rate of $\score$.
Initially, set $V':=\emptyset$.
For each content $v$ stored in the cache $c$ by this caching allocation, add the vertex $v$ to $V'$.
Then, $V'$ contains at most $k$ vertices.
Since this caching allocation in $G'$ achieves a cache hit rate of $\score$, there are at least $t$ users in $G'$ for which both of the contents that they requested with non-zero probability are stored in adjacent caches.
In particular, at least one of the contents they requested with non-zero probability is stored in the cache $c$.
Since each of those users in $G'$ corresponds to an edge in $G$ containing both of the vertices that correspond to the contents requested with non-zero probability by that user, the set $V'$ satisfies the desired property in $G$.
That is, $|V'|\leq k$ and at least $t$ edges in $G$ contain a vertex from $V'$.
\end{proof} 

The previous theorem rules out tractability even on planar networks, but does not do so when $\Delta$ and $\mcap$ are constants.
In our next result, we rule out tractability for \prob, even when restricted to planar networks where $\Delta$, $\like$, and $\mcap$ are all fixed constants.
We establish this via a reduction from the \NP-hard {\sc Planar 3-SAT-E3} problem~\citep{MP93}, whose definition is as follows.

\defproblem{{\sc Planar 3-SAT-E3}}
{A 3-CNF formula $\phi$ in which each clause contains $2$ or $3$ literals, each variable appears in exactly $3$ clauses, and its incidence graph is planar.}
{Is there a satisfying assignment for $\phi$?}

\begin{restatable}{theorem}{thmsatred}
\label{thm:sat-red}
\prob\ is \NP-hard, even if $G$ is planar, $\mcap=1$, $\like=3$, and $\Delta=5$.
\end{restatable}

  \begin{proof}
We reduce from {\sc Planar 3-SAT-E3}.
From an instance $\phi$ of {\sc Planar 3-SAT-E3}, we construct an instance $(G,\score)$ of \prob\ as follows.
Let $x_1,\ldots,x_n$ and $C_1,\ldots,C_m$ be the variables and clauses in $\phi$, respectively.
For each of the literals $x_1,\overline{x}_1,\ldots,x_n,\overline{x}_n$ in $\phi$, there is a content in the catalog with the same name.
We construct $G$ from the planar incidence graph $G_{\phi}$ of $\phi$ as follows.
Each variable vertex $x_i$, $i\in [n],$ in $G_{\phi}$ corresponds to a cache $c_{x_i}$ of capacity $1$ in $G$.
Each clause vertex $C_j$, $j\in [m]$, in $G_{\phi}$ corresponds to a user $u_{C_j}$ (of weight $1$) in $G$ that only requests (with non-zero probability) the contents corresponding to the literals that the clause $C_j$ contains, and $u_{C_j}$ requests these contents with equal probability.
Thus, if a variable appears in a clause in $\phi$, then the variable's corresponding cache is adjacent to the clause's corresponding user in $G$.
Further, for each cache $c_{x_i}$, $i\in [n]$, there is an additional cache $c_{x_i}'$ and two additional users $u_{x_i}$ and $u_{\overline{x}_i}$.
The user $u_{x_i}$ ($u_{\overline{x}_i}$, respectively) only requests the content $x_i$ ($\overline{x}_i$, respectively) with non-zero probability, i.e., $p_{u_{x_i}x_i}=1$ ($p_{u_{\overline{x}_i}\overline{x}_i}=1$, respectively).
Both the caches $c_{x_i}$ and $c_{x_i}'$ are adjacent to both the users $u_{x_i}$ and $u_{\overline{x}_i}$.
Since these are all the adjacencies of the cache $c_{x_i}'$ and the users $u_{x_i}$ and $u_{\overline{x}_i}$, their respective vertices can be placed arbitrarily close to $c_{x_i}$, thus maintaining that $G$ is planar.
Also, for each user $u_{C_j}$, $j\in [m]$, such that $C_j$ contains three literals in $\phi$, there are two additional caches $c_{u_{C_j}}$ and $c_{u_{C_j}}'$ of capacity $1$ that are adjacent only to $u_{C_j}$.
For each user $u_{C_j}$, $j\in [m]$, such that $C_j$ contains only two literals in $\phi$, there is one additional cache $c_{u_{C_j}}$ of capacity $1$ that is adjacent only to $u_{C_j}$.
Similarly, these are all the adjacencies of these additional caches, and so, their respective vertices can be placed arbitrarily close to $u_{C_j}$, thus maintaining that $G$ is planar.
\begin{figure}
	\centering  
\includegraphics[scale=0.55]{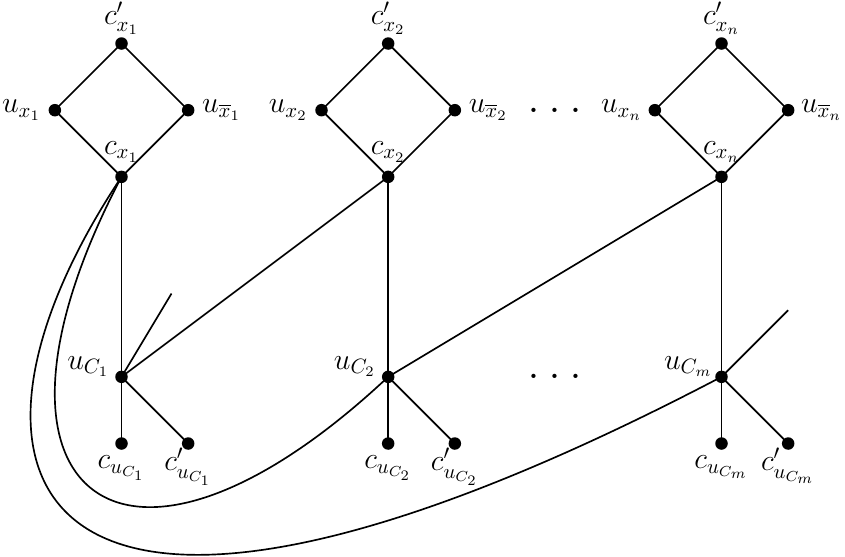}
\caption{Illustration of the planar graph $G$ constructed from an instance $\phi$ of {\sc Planar 3-SAT-E3} in the proof of Theorem~\ref{thm:sat-red}.
Here, $\phi$ contains the clause $C_1$ containing the variables $x_1$, $x_2$, and a third arbitrary one (denoted by a line protruding from the vertex $u_{C_1}$), the clause $C_2$ containing the variables $x_1$, $x_2$, and $x_n$, and the clause $C_m$ containing the variables $x_1$, $x_n$, and a third arbitrary one.}\label{fig:planar}
\end{figure}
Lastly, set $\score:=\users$.
This completes the construction of $(G,\score)$, which is clearly achieved in polynomial time.
Further, it is easy to verify that, in $(G,\score)$, $G$ is planar, $\mcap=1$, $\like=3$, and $\Delta=5$ (since each variable appears in $3$ clauses in $\phi$).
See Figure~\ref{fig:planar} for an illustration of $G$.

For the first direction, assume that there is a satisfying assignment for $\phi$.
For each $i\in [n]$, if the satisfying assignment for $\phi$ sets the variable $x_i$ to True (False, respectively), then store the content $x_i$ ($\overline{x}_i$, respectively) in the cache $c_{x_i}$ in $G$.
For each $i\in [n]$, if the content stored in the cache $c_{x_i}$ according to the above caching allocation is $x_i$ ($\overline{x}_i$, respectively), then store the content $\overline{x}_i$ ($x_i$, respectively) in the cache $c_{x_i}'$ in $G$.
Then, the current caching allocation in $G$ ensures that, for all $i\in [n]$, the user $u_{x_i}$ ($u_{\overline{x}_i}$, respectively) is adjacent to a cache that stored the content $x_i$ ($\overline{x}_i$, respectively).
Further, since this caching allocation corresponds to a satisfying assignment for $\phi$, it ensures that, for all $j\in [m]$, the user $u_{C_j}$ is adjacent to a cache that stored one of the contents corresponding to one of the literals that the clause $C_j$ contains in $\phi$.
Indeed, each clause $C_j$, $j\in [m]$, in $\phi$ corresponds to a user $u_{C_j}$ in $G$ that is adjacent to $2$ or $3$ caches that correspond to the $2$ or $3$ variables contained in $C_j$ in $\phi$.
To complete the caching allocation, for each $j\in [m]$, for the at most $2$ remaining contents (at most $1$ if $C_j$ contains only two literals in $\phi$) requested with non-zero probability by the user $u_{C_j}$ that may not already be stored in one of its adjacent caches of the form $c_{x_i}$, it suffices to store one of them in the cache $c_{u_{C_j}}$ and the other in the cache $c_{u_{C_j}}'$ (if it exists and if needed).
Thus, this caching allocation in $G$ achieves a cache hit rate of $\score$, and it is easy to check that it is feasible.

For the other direction, assume that there is a feasible caching allocation in $G$ achieving a cache hit rate of $\score$.
Since it achieves a cache hit rate of $\score$, this caching allocation ensures that, for each $i\in [n]$, the cache $c_{x_i}$ stored either the content $x_i$ or the content $\overline{x}_i$ since $p_{u_{x_i}x_i}=1$ and $p_{u_{\overline{x}_i}\overline{x}_i}=1$, and the users $u_{x_i}$ and $u_{\overline{x}_i}$ are only adjacent to one other cache of capacity $1$, namely $c_{x_i}'$.
For each $i\in [n]$, if this caching allocation in $G$ stores the content $x_i$ ($\overline{x}_i$, respectively) in the cache $c_{x_i}$, then set the variable $x_i$ to True (False, respectively).
We now argue that this results in a satisfying assignment for $\phi$.
Recall that, for each $j\in [m]$, the user $u_{C_j}$ either has $1$ (if $C_j$ contains only two literals in $\phi$) or $2$ (if $C_j$ contains three literals in $\phi$) additional caches adjacent to it that are not of the form $c_{x_i}$.
Hence, this caching allocation ensures that, for each $j\in [m]$, one of the contents requested with non-zero probability by the user $u_{C_j}$ is stored in a cache of the form $c_{x_i}$ adjacent to $u_{C_j}$.
As each of these contents corresponds to one of the literals contained in the clause $C_j$ in $\phi$, the above truth assignment for $\phi$ is a satisfying one.
\end{proof} 

Neither of the previous results rules out tractability for \prob\ when parameterized by the vertex cover number of the graph, which is the minimum size of a subset~$S$ of the vertices such that every edge is incident to at least one vertex in $S$.
In the next section dedicated to unifying the remaining open cases, we show that, for \prob, this parameterization is complexity-theoretically equivalent to several other ones we already considered. 
As a by-product, this yields \XP-tractability for \prob\ parameterized by the vertex cover number, contrasting the lower bounds ruling out the use of essentially all other structural graph parameters.

\section{Interreducibility: Linking the Open Cases}
\label{sec:interreduce}

For several of the studied cases, we obtained \XP\ algorithms, but lack the corresponding \W[1]-hardness proofs which would rule out inclusion in \FPT. 
In this section, for all parameterizations of \prob\ which we show to admit \XP\ algorithms in Section~\ref{sec:upper-bounds} (in particular $\caches$, $\users$, $\caches+\users$, $\caches+\like$, and $\users+\mcap$) and also for the vertex cover number parameter mentioned at the end of the previous section, we prove the following: either \prob\ is \W\textup{[1]}-hard for each of these, or it is \FPT\ for each of these.
That is, all of the arising parameterized problems are equivalent and there is only a single open case left for \prob. 
We emphasize that these results are not trivial, as here it does not hold that bounding one parameter is immediately equivalent to bounding the other.

\begin{restatable}{theorem}{fptequiv}
\label{thm:fpt-equivalent}
Let $\varkappa$ and $\varkappa'$ denote any two of the following six parameters: $\caches$, $\users$, $\users+\mcap$, $\caches+\users$, $\caches+\like$, and the vertex cover number $\vc(G)$.
Then, there exists a parameterized reduction from \prob\ parameterized by $\varkappa$ to \prob\ parameterized by $\varkappa'$. 
\end{restatable}

\begin{proof}
First, note that we do not need to present parameterized reductions in the cases where $\varkappa\geq \varkappa'$ as the trivial reduction where the two instances are identical is valid in these cases.
In particular, this implies that we do not need to consider the case where $\varkappa=\caches+\users$ and $\varkappa'=\vc(G)$.
We now provide a case analysis that can easily be checked to be exhaustive as parameterized reductions respect transitivity.
Also, since each of the following reductions takes polynomial time and the equivalence of the two instances is immediate, we simply present the reductions.

\noindent{\bf Case 1:} $\varkappa=\users$ and $\varkappa'=\caches$.
This case was covered in the proof of Observation~\ref{obs:users}.

\noindent {\bf Case 2:} $\varkappa=\caches$ and $\varkappa'=\users$.
From an instance $(G,\score)$ of \prob\ with $\caches$ caches, we obtain an equivalent instance $(G',\score)$ of \prob\ with $\users'\leq 2^{\caches}$ users as follows.
For each set $U^* \subseteq \mathcal{\users}$ of $2$ or more users with the same neighborhood in $G$, delete all the users except for one, denote this user by $u^*$, and set the weight and content request probabilities of $u^*$ in such a way that any potential cache hits for $U^*$ in $G$ would result in the same amount of cache hits for $u^*$ in $G'$.
Specifically, $w(u^*) := \sum_{u\in U^*} w(u)$ and 
$p_{u^*s}:=\left(\sum_{u\in U^*} w(u) p_{us}\right)/w(u^*)$ for each $s\in \mathcal{\catalog}$.
Since there are $2^{\caches}$ different subsets of the caches in $G$, there are at most $2^{\caches}$ users in $G'$.

\noindent{\bf Case 3:} $\varkappa=\users$ and $\varkappa'=\users+\mcap$.
From an instance $(G,\score)$ of \prob\ with $\users$ users, we obtain an equivalent instance $(G',\score)$ of \prob\ with a maximum cache capacity of $\mcap'=1$ and $\users$ users as follows. Let $\caches$ be the number of caches in $(G,\score)$.
For each $i\in [\caches]$, replace the cache $c_i$ in $G$ by $\capa(c_i)$ caches of capacity~$1$ (recall that $\capa(c_i)\leq \catalog$) with the same neighborhood as $c_i$.

\noindent{\bf Case 4:} $\varkappa=\caches$ and $\varkappa'=\caches+\like$.
From an instance $(G,\score)$ of \prob\ with $\caches$ caches, we obtain an equivalent instance $(G',\score)$ of \prob\ with $\caches$ caches in which each user requests only a single content with probability $1$ as follows.
For each user $u$ in $G$, delete $u$, and, for each content $s$ requested with non-zero probability by $u$, add a user $u_s$ with weight $w(u)\cdot p_{us}$ such that $p_{u_ss}=1$ and $u_s$ has the same neighborhood as $u$ had.

\noindent{\bf Case 5:} $\varkappa=\vc(G)$ and $\varkappa'=\caches+\users$.
From an instance $(G,\score)$ of \prob\ with vertex cover number $\vc(G)$, we obtain an equivalent instance $(G',\score)$ of \prob\ where the number of caches and users are both upper-bounded by $\vc(G)$ as follows. Let $\caches$ and $\users$ be the number of caches and users in $(G,\score)$, respectively, and let $\mathcal{\users}$ be the set of users in $(G,\score)$. We first use the simple polynomial-time 2-approximation algorithm for {\sc Vertex Cover} to compute a vertex cover $X$ of $G$ such that $X\leq 2\cdot \vc(G)$, and let $I=V(G)\setminus X$.
By the definition of $X$, $I$ is an independent set.
For each set $U^* \subseteq \mathcal{\users}\cap I$ of $2$ or more users with the same neighborhood in $G$, delete all the users except for one, denote this user by $u^*$, and set the weight and content request probabilities of $u^*$ as in Case 2.
For each set of $2$ or more caches in $I$ with the same neighborhood in $G$, delete all the caches except for one, and give this cache a capacity equal to the sum of the capacities of those deleted caches plus its original capacity.
As there are at most $2^{|X|}$ distinct subsets of users and caches in $X$, there are at most $2^{|X|}+|X|\leq 4^{\vc(G)}+2\cdot\vc(G)$ caches and users in $G'$.
\end{proof} 

We conjecture that \prob\ is \W[1]-hard under these parameterizations, but believe that a proof requires novel techniques or insights into the problem. Moreover, as \hetprobun\ generalizes \prob, resolving this conjecture in the affirmative would also resolve the sole open case for \hetprobun. 

\section{Generalizations, Impact, and Conclusion}

The \FPT\ algorithms developed in Theorems~\ref{thm:C+S}~and~\ref{thm:xp-c} can be easily adapted to handle a more general framework associated with {\sc Network-Caching} since they consider all the relevant feasible caching allocations.
Indeed, for any objective function that can be computed in the desired \FPT~time when given a caching allocation, these algorithms can also compute the optimal value of that objective function along with its associated caching allocation.
Moreover, most objective functions in the literature satisfy the above condition.
For example, these algorithms can be trivially modified to deal with variants of {\sc Network-Caching} where weights are added to the edges of the bipartite graph which represent the caching gain obtained from retrieving a content requested by the user from a specific cache~\citep{ioannidis2016adaptive,tsigkari2022approximation}, and/or the objective function concerns other metrics such as QoS, streaming rate or energy consumption~\citep{paschos2020cache}.
Our hardness results also carry over to these variants as well as generalizations combining caching with other network-related decisions~\citep{dehghan2016complexity,krolikowski2018decomposition,ricardo2021caching}.
On the other hand, these hardness results cannot be lifted to non-discrete variants of \textsc{Network-Caching}, as these can typically be solved in polynomial time~\citep{femto}.

All of our algorithms are deterministic and implementable in CDNs or inference delivery networks, and 
in fact it is reasonable to expect some of the studied parameters to achieve small values in practice (see also Section~\ref{sec:prelims}). However, 
we believe that a natural next step would be to find the theoretically fastest algorithms under the (Strong) Exponential Time Hypothesis.
Further, designing informed heuristics based on our complexity analysis---as was successfully done in other fields~\citep{backstrom2012complexity,rost2019parametrized,komusiewicz2023group}---would be an interesting alternate direction one could take.

\section{Acknowledgements}
This work received funding from the Austrian Science Fund (FWF)~[10.55776/Y1329 and 10.55776/COE12], the WWTF Vienna Science and Technology Fund (Project 10.47379/ICT22029),
the Spanish Ministry of Economic Affairs and Digital Transformation and the European Union-NextGenerationEU through the  project 6G-RIEMANN~(TSI-063000-2021-147), the EU Horizon Europe TaRDIS project~(grant agreement 101093006), and the Smart Networks and Services Joint Undertaking (SNS JU) under the European Union's Horizon Europe and innovation programme under Grant Agreement No 101139067 (ELASTIC). Views and opinions expressed are
however those of the author(s) only and do not necessarily reflect those of the European
Union. Neither the European Union nor the granting authority can be held responsible for them.

\bibliographystyle{apalike}
\bibliography{shortened_bib}

\end{document}